\UseRawInputEncoding
\documentclass[10pt]{article}
\usepackage{amsfonts}
\usepackage{amsfonts}
\usepackage{amsfonts}
\usepackage{amsfonts}
\usepackage{amsfonts}
\usepackage{amsfonts}
\usepackage{amsfonts}
\usepackage{amsfonts}
\usepackage{amsfonts}
\usepackage{amsfonts}
\usepackage{amsfonts}
\usepackage{amsfonts}
\usepackage{amsfonts}
\usepackage{amsfonts}
\usepackage{amsfonts}
\usepackage{amsfonts}
\usepackage{amsfonts}
\usepackage{amsfonts}
\usepackage{amsfonts}
\usepackage{amsfonts}
\usepackage{amsfonts}
\usepackage{amsfonts}
\usepackage{amsfonts}
\usepackage{amsfonts}
\usepackage{amsfonts}
\usepackage{amsfonts}
\usepackage{amsfonts}
\usepackage{amsfonts}
\usepackage{amsfonts}
\usepackage{amsfonts}
\usepackage{amsfonts}
\usepackage{amsfonts}
\usepackage{amsfonts}
\usepackage{amsfonts}
\usepackage{amsfonts}
\usepackage{amsfonts}
\usepackage{amsfonts}
\usepackage{amsfonts}
\usepackage{amsfonts}
\usepackage{amsfonts}
\usepackage{amsfonts}
\usepackage{amsfonts}
\usepackage{amsfonts}
\usepackage{amsmath}
\usepackage{amssymb, amscd, amsthm}
\usepackage[all]{xy}
\usepackage[dvips]{graphicx}
\usepackage{verbatim}
\usepackage{cite}
\usepackage{enumerate}
\usepackage{float}
\usepackage{diagbox}

\newtheorem{lemma}{\textbf{Lemma}}[section]
\newtheorem{theorem}{\textbf{Theorem}}
\newtheorem{remark}{\textbf{Remark}}[section]
\newtheorem{corollary}{\textbf{Corollary}}[section]

\newtheorem{example}{\textbf{Example}}[section]
\newtheorem{proposition}{\textbf{Proposition}}[section]

\newtheorem{Definition}{\textbf{Definition}}
\usepackage{longtable}

\usepackage{multirow}
\usepackage{tikz}

\begin{document}

\baselineskip 17pt\title{\Large\bf Groups of linear isometries on weighted poset block spaces}
\author{\large  Wen Ma \quad\quad Jinquan Luo\footnote{The authors are with School of Mathematics
and Statistics \& Hubei Key Laboratory of Mathematical Sciences, Central China Normal University, Wuhan China.\newline
 E-mails: mawen95@126.com(W.Ma),  luojinquan@mail.ccnu.edu.cn(J.Luo)}}
\date{}
\maketitle

{\bf Abstract}: In this paper, we introduce a new family of metrics, weighted poset block metric, that combine the weighted coordinates poset metric introduced by Panek et al. [(\ref{panek})] and the metric for linear error-block codes introduced by Feng et al. [(\ref{FENG})]. This type of metrics include many classical metrics such as Hamming metric, Lee metric, poset metric, pomset metric, poset block metric, pomset block metric and so on. We give a complete description of the groups of linear isometries of these metric spaces in terms of a semi-direct product. Furthermore, we obtain a Singleton type bound for codes equipped with weighted poset block metric and define MDS codes. As a special case when the poset is a chain, we show that MDS codes are equivalent to perfect codes.

{\bf Key words}: linear isometry, poset block metric, principal ideal, pomset block metric, weighted poset block metric

\section{Introduction}

\quad\;The study of codes endowed with a metric other than the Hamming metric gained momentum since 1990's, such as the \emph{poset metric} by Brualdi et al. ([\ref{BG}]). Over the last two decades, the study of codes in the poset metric has made the subject of coding theory to see several developments. It paved a way for studying codes equipped with various metrics.

Feng, Xu and Hickernell ([\ref{FENG}]) introduced the block metric by partitioning the set of coordinate positions of $\mathbb{F}_q^n$ and studied MDS block codes. \emph{Poset block metric} was introduced by Alves, Panek, and Firer ([\ref{ALVES}]) unifying poset metric and block metric. Later, Dass, Sharma and Verma obtained a Singleton type bound for poset block codes and define a maximum distance separable poset block code as a code meeting this bound. They also extended the concept of $I$-balls to poset block metric and describe $r$-perfect and MDS $(P,\pi)$-codes in terms of $I$-perfect codes. \emph{Niederreiter-Rosenbloom-Tsfasman block metric} (in short, \emph{NRT block metric}) is a particular case of poset block when the poset is a chain. Panek, Firer and Alves in [\ref{PANEK}] classified the classes of equivalent codes and developed much of the classical theory for NRT block codes.

As the support of a vector $v$ in $\mathbb{F}_q^n$ is a set and hence induces order ideals and metrics on $\mathbb{F}_q^n$, the poset metric codes could not accommodate Lee metric structure due to the fact that the support of a vector with respect to Lee weight is not a set but rather a multiset. In order to handle Lee metric, a much general class of metrics called \emph{pomset metric} is introduced by Irrinki and Selvaraj ([\ref{IGS1}],[\ref{IGS2}],[\ref{IGS3}]) for codes over $\mathbb{Z}_m$. Furthermore, construction of pomset codes are obtained and their metric properties like minimum distance, covering radius, MacWilliams' identity, Singleton-type bound, perfect codes and so on are determined.

The group of linear isometries were determined for the Rosenbloom-Tsfasman space, generalized Rosenbloom-Tsfasman space and crown space (see [\ref{CROWN}],[\ref{GRT}],[\ref{RT}]). In [\ref{FIRERHYUN}], Panek et al. gave a complete description of group of linear isometries of poset metric space. Extending their approach,
Panek and Pinheiro studied \emph{Weighted coordinates poset metric} ([\ref{panek}]) that include any additive metric space such as Hamming metric and Lee metric, as well as the poset metric and pomset metric. They provided a complete description of the groups of linear isometries of these metric spaces in terms of a semidirect product, which turns out to be similar to the case of poset metric spaces.

It is natural to ask: does there exist a metric that will be compatible with additive metric and block metric? In this work, we combine weighted poset metric with error-block metric to obtain a further generalization called the\emph{ weighted poset block metric} which includes not only all additive metrics mentioned above but also some block metric such as \emph{poset block metric}, \emph{pomset block metric} and so on.

The paper is organized as follows. Section 2 contains basic notions of posets and defines the weighted poset block metric over $\mathbb{F}_q^n$. In Section 3, we give a complete description of groups of linear isometries of the metric space $\left(V,d_{w,(P,\pi)}\right)$, for any weighted poset block metric $d_{w,(P,\pi)}$. In Section 4, we give some examples to illustrate our conclusion. Finally, we establish a singleton type bound for codes with weighted poset block metric and define a maximum distance separable $(P,\pi,w)$-code (MDS $(P,\pi,w)$-code) as a code meeting this bound. The connection of MDS $(P,\pi,w)$-codes with perfect codes is also investigated when $P$ is considered to be a chain.

\usetikzlibrary{shapes.geometric, arrows}
\thispagestyle{empty}
\tikzstyle{startstop} = [rectangle, rounded corners, minimum width = 2cm, minimum height=1cm,text centered, draw = black]
\tikzstyle{io} = [rectangle, rounded corners, minimum width = 2cm, minimum height=1cm,text centered, draw = black]
\tikzstyle{process} = [rectangle, rounded corners, minimum width = 2cm, minimum height=1cm,text centered, draw = black]
\tikzstyle{decision} = [rectangle, rounded corners, minimum width = 2cm, minimum height=1cm,text centered, draw = black]
\tikzstyle{arrow} = [->,>=stealth]
\begin{center}
\begin{table}[H]
    \caption{}
\begin{center}
\begin{tikzpicture}[node distance=1cm]

\node[startstop](start){weighted poset block metric};
\node[io, below of = start, yshift = -1cm](in1){weighted coordinates poset metric};
\node[process, left of = in1, xshift = -4cm] (pro1) {poset block metric};
\node[process, right of = in1, xshift = 4cm] (pro2) {pomset block metric};
\node[io, below of = in1, yshift = -1cm](in2){additive metric};
\node[process, left of = in2, xshift = -4cm] (pro3) {poset metric};
\node[process, left of = in2, xshift = -7cm] (pro7) {NRT block metric};
\node[process, right of = in2, xshift = 4cm] (pro4) {pomset metric};
\node[process, below of = pro3, yshift = -1cm] (pro5) {Hamming metric};
\node[process, below of = pro4, yshift = -1cm] (pro6) {Lee metric};
\node[io, below of = in2, yshift = -3cm](in3){Hamming and Lee metrics on $\mathbb{F}_2$ and $\mathbb{F}_3$};
\coordinate (point1) at (-2cm, -4cm);
\draw [arrow] (start) -- node [right] {$\pi(i)=1$} (in1);
\draw [arrow] (start) -- (pro1);
\draw [arrow] (start) -- (pro2);
\draw [arrow] (in1) -- node [right] {anti-chain} (in2);
\draw [arrow] (in1) -- (pro3);
\draw [arrow] (in1) -- (pro4);
\draw [arrow] (pro2) -- node [right] {$\pi(i)=1$} (pro4);
\draw [arrow] (pro1) -- node [right] {$\pi(i)=1$} (pro3);
\draw [arrow] (pro1) -- node [left] {chain} (pro7);
\draw [arrow] (pro3) -- node [right] {anti-chain} (pro5);
\draw [arrow] (pro4) -- node [right] {anti-chain} (pro6);
\draw [arrow] (in2) -- (pro6);
\draw [arrow] (in2) -- (pro5);
\draw [arrow] (pro5) -- (in3);
\draw [arrow] (pro6) -- (in3);
\end{tikzpicture}
\end{center}
\end{table}

\end{center}

\section{Preliminaries}

\quad\; In this section, we give basic definitions and results of a weighted poset block metric to be used in the subsequent sections.

\subsection{On weighted poset block metric}

\quad\; The definitions of weight and metric can be defined on general rings. In particular, we restrict it to finite fields because it is the most explored in the context of coding theory.

Let $\mathbb{F}_q$ be the finite field of order $q$ and $\mathbb{F}_q^n$ the $n$-dimensional vector space over $\mathbb{F}_q$.

\begin{Definition}
A map $d: \mathbb{F}_q^n\times \mathbb{F}_q^n\rightarrow \mathbb{N}$ is a metric on $\mathbb{F}_q^n$ if it satisfies the following conditions:
\begin{enumerate}[(1)]
\item (non-negativity) $d(u,v)\geq 0$ for all $u,v\in\mathbb{F}_q^n$ and $d(u,v)=0$ if and only if $u=v$;
\item (symmetry) $d(u,v)=d(v,u)$ for all $u,v\in \mathbb{F}_q^n$;
\item (triangle inequality) $d(u,v)\leq d(u,w)+d(w,v)$ for all $u,v,w\in \mathbb{F}_q^n$.
\end{enumerate}
\end{Definition}

\begin{Definition}
A map $w: \mathbb{F}_q^n\rightarrow \mathbb{N}$ is a weight on $\mathbb{F}_q^n$ if it satisfies the following conditions:
\begin{enumerate}[(1)]
\item $w(u)\geq 0$ for all $u\in \mathbb{F}_q^n$ and $w(u)=0$ if and only if $u=0$;
\item $w(u)=w(-u)$ for all $u\in \mathbb{F}_q^n$;
\item $w(u+v)\leq w(u)+w(v)$ for all $u,v\in\mathbb{F}_q^n$.
\end{enumerate}
\end{Definition}

It is straightforward to prove that, if $w$ is a weight over $\mathbb{F}_q^n$, then the map $d_w$ defined by $d(u,v)=w(u-v)$ is a metric on $\mathbb{F}_q^n$. See [\ref{DEZA}] and [\ref{GAB}] for detailed discussion on weight and metric.

Let $P$ be a set. A \emph{partial order} on $P$ is a binary relation $\leq$ on $P$ such  that for all $x,y,z\in P$, we have $x\leq x$ (\emph{reflexivity}), $x\leq y$ and $y\leq x$ imply $x=y$ (\emph{antisymmetry}), $x\leq y$ and $y\leq z$ imply $x\leq z$ (\emph{transitivity}). A set $P$ equipped with an order relation $\leq$ is said to be a \emph{poset}.  An element $a\in P$ is a \emph{maximal element} of $P$ if $a\leq b$ and $b\in P$ imply $a=b$. We denote by $\max P$ the set of all maximal elements of $P$. A poset $P$ is a \emph{chain} if any two elements of $P$ are comparable. The opposite of a chain is an antichain, that is, a poset $P$ is an \emph{antichain} if $x\leq y$ in $P$ only when $x=y$. We call a subset $Q$ of $P$ an \emph{ideal} if, whenever $x\in Q$, $y\in P$ and $y\leq x$, we have $y\in Q$. For a subset $E$ of $P$, the \emph{ideal generated by $E$}, denoted by $\langle E\rangle$, is the smallest ideal of $P$ containing $E$. We prefer to denote the ideal generated by $\{i\}$ as $\langle i\rangle$ instead of $\langle\{i\}\rangle$. We denoted by $\langle i\rangle^{*}$ the \emph{difference} $\langle i\rangle-\{i\}=\{j\in P: j<i\}$.

Given two posets $P$ and $Q$, we say that $P$ and $Q$ are \emph{isomorphic}, and write $P\cong Q$, if there exists a bijective map $\varphi$ from $P$ onto $Q$ such that $x\leq y$ in $P$ if and only if $\varphi(x)\leq \varphi(y)$ in $Q$. Then $\varphi$ is called an \emph{order-isomorphism}.  An order-isomorphism $\varphi:P\rightarrow P$ is called an \emph{automorphism} and we denote by $Aut(P)$ the group of automorphisms of $P$.

Let $[s]=\{1,2,\ldots,s\}$. Let $P=([s],\leq)$ be a poset and let $\pi: [s]\rightarrow\mathbb{N}^{+}$ be a map such that $n=\sum\limits_{i=1}^s\pi(i)$. The map $\pi$ is said to be a \emph{labeling} of the poset $P$, and the pair $(P,\pi)$ is called a \emph{poset block structure} over $[s]$. Denote $\pi(i)$ by $k_i$. We take $V_i$ as the $\mathbb{F}_q$-vector space $\mathbb{F}_q^{k_i}$ for all $1\leq i\leq s$. We define $V$ as the direct sum
\begin{equation}\label{sum}
V=V_1\oplus V_2\oplus\cdots\oplus V_s
\end{equation}
which is isomorphic to $\mathbb{F}_q^n$. Each $u\in V$ can be uniquely decomposed as
\begin{equation}\label{decom}
u=u_1+u_2+\cdots+u_s
\end{equation}
 where $u_i=(u_{i_1},\ldots,u_{ik_i})\in V_i$ for $1\leq i\leq s$.

Let $w$ be a weight on $\mathbb{F}_q$ and $P=([s],\leq)$ be a poset. Given $u\in V$, set
\begin{equation}\label{local}
W_i(u)=\max\left\{w(u_{ij}): 1\leq j\leq k_i\right\}\ \text{for}\ 1\leq i\leq s;
\end{equation}

\begin{equation}\label{max}
M_w=\max\left\{w(\alpha): \alpha\in\mathbb{F}_q\right\};
\end{equation}
\begin{equation}\label{min}
m_w=\min\left\{w(\alpha): 0\neq\alpha\in\mathbb{F}_q\right\}.
\end{equation}

The \emph{block support} or \emph{$\pi$-support} of $u\in V$ is the set
$$supp_{\pi}(u)=\left\{i\in [s]:u_i\neq 0\right\}.$$
We denote by $I_u^P$ the ideal generated by $supp_{\pi}(u)$ and denote by $M_u^P$ the set of maximal elements in the ideal $I_u^P$. The \emph{$(P,\pi,w)$-weight} of $u$ is defined as
\begin{equation}
\overline{\omega}_{w,(P,\pi)}(u)=\sum\limits_{i\in M_u^P} W_i(u)+\sum\limits_{i\in I_u^P\setminus M_u^P}M_w.
\end{equation}
For $u,v\in V$, define their \emph{$(P,\pi,w)$-distance} as
\begin{equation}
d_{w,(P,\pi)}(u,v)=\overline{\omega}_{w,(P,\pi)}(u-v).
\end{equation}

\begin{theorem}
The $(P,\pi,w)$-weight defined above is a weight over $V$ and thus the $(P,\pi,w)$-distance $d_{w,(P,\pi)}(.,.)$ is a metric over $V$.
\end{theorem}

\begin{proof}
\begin{enumerate}[(1)]
\item Clearly $\overline{\omega}_{w,(P,\pi)}(u)\geq 0$ and $\overline{\omega}_{w,(P,\pi)}(u)=0$ iff $u=0$ for all $u\in V$.
\item As $w(a)=w(-a)$ for all $a\in\mathbb{F}_q$ and $supp_{\pi}(u)=supp_{\pi}(-u)$ for all $u\in V$, it follows that $\overline{\omega}_{w,(P,\pi)}(u)=\overline{\omega}_{w,(P,\pi)}(-u)$.
\item  Finally we show that the $(P,\pi)$-weight satisfies the triangle inequality. Take $u,v\in V$, we have
\begin{eqnarray*}
\overline{\omega}_{w,(P,\pi)}(u+v)&=&\sum\limits_{i\in M_{u+v}^P}W_i(u+v)+\sum\limits_{i\in I_{u+v}^P\setminus M_{u+v}^P} M_w\\
&\leq&\sum\limits_{i\in M_{u+v}^P}[W_i(u)+W_i(v)]+\sum\limits_{i\in I_{u+v}^P\setminus
M_{u+v}^P} M_w\\
&\leq&\sum\limits_{i\in M_u^P}W_i(u)+\sum\limits_{i\in I_u^p\setminus M_u^P}M_w+\sum\limits_{i\in M_v^P}W_i(v)+\sum\limits_{i\in I_v^p\setminus M_v^P}M_w\\
&=&\overline{\omega}_{w,(P,\pi)}(u)+\overline{\omega}_{w,(P,\pi)}(v).
\end{eqnarray*}
The second and third inequalities follow from that $w$ is a weight over $\mathbb{F}_q$ and $supp_{\pi}(u+v)\subseteq supp_{\pi}(u)\cup supp_{\pi}(v)$, $I_{u+v}^P\setminus M_{u+v}^P\subseteq I_u^P\setminus M_u^P\cup I_v^P\setminus M_v^P$ (see [\ref{panek}]) respectively. Therefore $\overline{\omega}_{w,(P,\pi)}$ is a weight over $V$.
\end{enumerate}

\end{proof}

The metric $d_{w,(P,\pi)}(.,.)$ is called the \emph{weighted poset block metric} and the pair $\left(V,d_{w,(P,\pi)}\right)$ is said to be a \emph{weighted poset block space}. When the label $\pi$ satisfies $\pi(i)=1$ for all $i\in [s]$ the weighted poset block metric is the weighted coordinates poset metric proposed by Panek et al. in [\ref{panek}] which combines and extends several classic metrics of coding theory such as Hamming metric, Lee metric, poset metric ([\ref{BG}]), pomset metric ([\ref{IGS1}]) and so on. We also refer the reader to [\ref{hyun}, \ref{RAM}, \ref{XKH}, \ref{xkh}] for a general metric called weighted poset metric. Let $P$ be a poset and $\theta:P\rightarrow \mathbb{R}^{+}$ be a map. For any $\beta\in V$, the weighted poset weight of $\beta$ is defined as
$$wt_{(P,w)}(\beta)=\sum\limits_{i\in\langle supp(\beta)\rangle_P}\theta(i).$$
Note that weighted poset metric is a metric which respect support condition. When the weight $w$ over $\mathbb{F}_q$ such that $w(\alpha)=t$ for all $\alpha\in\mathbb{F}_q$ and $\theta(i)=t$ for all $1\leq i\leq s$, weighted poset block metric would coincide with weighted poset metric over $V$.

\subsection{Special cases of weighted poset block metric}

\quad\;In this section, we will show some important metrics that can be deduced from weighted poset block metric.

\textbf{1. Poset block metric}

The \emph{poset block weight} is defined by
$$w_{(P,\pi)}(v)=\left|\langle supp_{\pi}(v)\rangle\right|$$
for all $v\in V$. For $u,v\in V$,
$$d_{(P,\pi)}(u,v)=w_{(P,\pi)}(u-v)$$
defines a metric over $V$ called the \emph{poset block metric} and the pair $\left(V,d_{(P,\pi)}\right)$ is said to be a \emph{poset block space}.
It is clear that poset block weight is a particular case of weighted poset block weight when $w$ is taken to be Hamming weight over $\mathbb{F}_q$.

\noindent \textbf{2. Pomset block metric}

Though we define weighted poset block metrics over $V=\bigoplus\limits_{i=1}^s\mathbb{F}_q^{k_i}$ in Section 2.1, it is still valid if we consider the free module over $\mathbb{Z}_m$ instead of $\mathbb{F}_q$.

Now we recall some definitions and properties regarding multisets needed for defining the pomset block metric and we show that it is a special case of weighted poset block metric.

Let $X$ be a set of elements. An \emph{multiset} (in short, \emph{mset}) drawn from $X$ is represented as $M=\{m_1/a_1,m_2/a_2,\ldots,m_n/a_n\}$ where $a_i\in X$ and $m_i$ is the number of occurrences of $a_i$ in $M$ denoted by $C_{M}(a_i)$. The \emph{cardinality} of $M$ is given by $|M|=\sum\limits_{a\in {X}}C_M(a)$. The set $M^{*}=\left\{a\in X: C_M(a)>0\right\}$ is called the \emph{root set} of $M$. A \emph{submultiset} (in short, \emph{subset}) of $M$ is an mset $M_1$ drawn from $X$ which satisfies $C_{M_1}(a)\leq C_M(a)$ for all $a\in X$.

Let $M_1$ and $M_2$ be two msets drawn from a set $X$. The \emph{union} of $M_1$ and $M_2$ is an mset $M$ such that for all $a\in X$, $C_{M}(a)=\max\left\{C_{M_1}(a), C_{M_2}(a)\right\}$. The \emph{Cartesian product} of $M_1$ and $M_2$, denoted by $M_1\times M_2$, is an mset $M=\{mm'/ab:m/a\in M_1,\ m'/b\in M_2\}$. A subset $R$ of $M_1$ and $M_2$ is said to be an mset \emph{relation} if every member $(m/a,m'/b)$ of $R$ has count $C_{M_1}(a)\cdot C_{M_2}(b)$.

An mset relation on an mset $M$ drawn from $X$ is said to be \emph{reflexive} iff $(m/a,m/a)\in R$ for all $m/a\in M$; \emph{anti-symmetric} iff $(m/a,m'/b)\in R$ and $(m'/b,m/a)\in R$ imply $m=m'$ and $a=b$; \emph{transitive} iff $(m/a,r/b)\in R$, $(r/b,t/c)\in R$ imply $(m/a,t/c)\in R$. The relation $R$ is called p\emph{omset relation} if it is reflexive, anti-symmetric and transitive. The pair $\mathbb{P}=(M,R)$ is called a \emph{pomset}. An \emph{order ideal} of $\mathbb{P}$ is a subset $I$ of $M$ such that if $m/a\in I$ and $(m'/b,m/a)\in R$ with $a\neq b$ imply $m'/b\in I$. An \emph{order ideal generated} by $m/a\in M$ is defined by
$$\langle m/a\rangle=\{m/a\}\cup \{m'/b\in M: (m'/b,m/a)\in R\ \text{for some}\ k>0\ \text{and}\ b\neq a\}.$$
An \emph{order ideal generated by a subset} $S$ of $M$ is defined by
 $\langle S\rangle=\bigcup\limits_{m/a\in S}\langle m/a\rangle$.

Consider the space $\mathbb{Z}_m^n$ and the multiset $M=\left\{\left\lfloor\frac{m}{2}\right\rfloor/1,\left\lfloor\frac{m}{2}\right\rfloor/2, \ldots,\left\lfloor\frac{m}{2}\right\rfloor/s\right\}$ drawn from the set $X=\{1,2,\ldots,s\}$. The \emph{Lee weight} of an element $a\in\mathbb{Z}_m$ is defined as $w_L(a)=\min\{a,m-a\}$ and the \emph{Lee block support} of $u\in\mathbb{Z}_m^{n}$ is defined as
$$supp_{(L,\pi)}(u)=\left\{s_i/i:s_i=w_{(L,\pi)}(u_i),s_i\neq 0\right\},$$
where
$$w_{(L,\pi)}(u_i)=\text{max}\left\{w_L(u_{i_t}):1\leq t\leq \pi(i)\right\}.$$
The \emph{$(\mathbb{P},\pi)$-weight} of $u\in\mathbb{Z}_m^{n}$ is given by
$$w_{(\mathbb{P},\pi)}(u)=\left|\langle supp_{(L,\pi)}(u)\rangle\right|.$$
For $u,v\in \mathbb{Z}_m^n$,
$$d_{(\mathbb{P},\pi)}(u,v)=w_{(\mathbb{P},\pi)}(u-v)$$
defines a metric over $\mathbb{Z}_m^n$ called pomset block metric. The pair $\left(\mathbb{Z}_m^n,d_{(\mathbb{P},\pi)}\right)$ is said to be a pomset block space.

\begin{proposition}
Let $\mathbb{P}=(M,R)$ be a pomset with $M=\left\{\left\lfloor\frac{m}{2}\right\rfloor/1,\left\lfloor\frac{m}{2}\right\rfloor/2, \ldots,\left\lfloor\frac{m}{2}\right\rfloor/s\right\}$. Let $P$ be the set $[s]$ and $\pi: [s]\rightarrow\mathbb{N}^{+}$ be a label with $\sum\limits_{i=1}^s\pi(i)=n$. Define a partial order on $P$ as
$$a\leq b\ \text{in}\ P\Leftrightarrow (r/a,t/b)\in R.$$
Then the $(P,\pi,w)$-weight and $(\mathbb{P},\pi)$-weight will coincide on $\mathbb{Z}_m^n$.
\end{proposition}

\begin{proof}
Take $u\in \mathbb{Z}_m^n$. Then $u$ can be written as $u=u_1+u_2+\cdots+u_s$ with $u_i\in\mathbb{Z}_m^{\pi(i)}$. Note that $$w_{(L,\pi)}(u_i)=\text{max}\left\{w_L(u_{i_t}):1\leq t\leq \pi(i)\right\}=W_i(u).$$
 Hence
\begin{eqnarray*}
w_{(\mathbb{P},\pi)}(u)&=&|\langle supp_{(L,\pi)}(u)\rangle|\\
&=&\sum\limits_{i\in M_u^{\mathcal {P}}}w_{(L,\pi)}(u_i)+\sum\limits_{i\in I_{u}^{\mathcal {P}}\setminus M_u^{\mathcal {P}}} \left\lfloor\frac{m}{2}\right\rfloor\\
&=&\sum\limits_{i\in M_u^{\mathcal {P}}}W_i(u)+\sum\limits_{i\in I_{u}^{\mathcal {P}}\setminus M_u^{\mathcal {P}}} \left\lfloor\frac{m}{2}\right\rfloor\\
&=&\overline{\omega}_{w,(P,\pi)}(u).
\end{eqnarray*}
\end{proof}

\begin{itemize}
  \item When the weight $w$ is the Hamming weight over $\mathbb{F}_q$, the $(P,\pi,w)$-weight is the poset block weight proposed by Alves et al. in [\ref{ALVES}]. In particular, poset block metric deduces NRT block metric when $P$ is taken to be a chain and deduces poset metric when $\pi(i)=1$ for all $i\in[s]$. Similarly, classical Hamming metric becomes a particular case of poset metric when $P$ is an anti-chain.
  \item When the weight $w$ is the Lee weight over $\mathbb{Z}_m$, the $(P,\pi,w)$-weight is the pomset block weight. In particular, pomset block metric deduces pomset metric when $\pi(i)=1$ for all $i\in [s]$. Particularly, pomset metric deduces Lee metric when $\mathbb{P}=(M,R)$ with $R=\left\{\left(\left\lfloor\frac{m}{2}\right\rfloor/a, \left\lfloor\frac{m}{2}\right\rfloor/a\right):1\leq a\leq s\right\}$.
\end{itemize}
The diagram 1 illustrates these facts.

\section{Linear isometries for weighted poset block spaces}

\quad\; In this section, we always assume that $w$ is a weight on $\mathbb{F}_q$, $P=([s],\leq)$ is a poset, $\pi: [s]\rightarrow \mathbb{N}$ is a labeling of the poset $P$ and $V=\bigoplus\limits_{i=1}^s\mathbb{F}_q^{k_i}$ which is isomorphic to $\mathbb{F}_q^n$.

\begin{Definition}
Let $\left(V,d_{w,(P,\pi)}\right)$ be a weighted poset block space. A linear isometry $T$ of $\left(V,d_{w,(P,\pi)}\right)$ is a linear transformation $T: V\rightarrow V$ such that
$$d_{w,(P,\pi)}\left(T(u),T(v)\right)=d_{w,(P,\pi)}(u,v)$$
for all $u,v\in V$. We also call a linear isometry as a $(P,\pi,w)$-isometry.
\end{Definition}

\begin{remark}
A linear transformation $T: V\rightarrow V$ is a linear isometry of $\left(V,d_{w,(P,\pi)}\right)$ if and only if $\overline{\omega}_{w,(P,\pi)}(T(u))=\overline{\omega}_{w,(P,\pi)}(u)$ for all $u\in V$.
\end{remark}

\begin{remark}
We claim that a $(P,\pi,w)$-isometry of $\left(V,d_{w,(P,\pi)}\right)$ is a bijection and its inverse is also a $(P,\pi,w)$-isometry. In fact, if $T$ is a $(P,\pi,w)$-isometry such that $T(u)=T(v)$ for $u,v\in V$ with $u\neq v$ then
$$d_{w,(P,\pi)}(u,v)=d_{w,(P,\pi)}(T(u),T(v))=0,$$
a contradiction. Therefore $T$ is injective. Since $V$ is a finite set, we have that $T$ is bijective. Let $T^{-1}$ be the inverse of $T$. For all $u,v\in V$, we have $$d_{w,(P,\pi)}\left(T^{-1}(u),T^{-1}(v)\right)= d_{w,(P,\pi)}\left(TT^{-1}(u),TT^{-1}(v)\right)=d_{w,(P,\pi)}(u,v).$$
\end{remark}

It follows from the above discussion that the set of all $(P,\pi)$-isometries of the weighted poset block space $\left(V,d_{w,(P,\pi)}\right)$ forms a group. We denoted it by $GL_{w,(P,\pi)}(V)$ and call it \emph{the group of linear isometries of} $\left(V,d_{w,(P,\pi)}\right)$.

We denote by $B=\{e_{11},\ldots,e_{1k_1}, e_{21},\ldots,e_{2k_2},\ldots,e_{s_1},\ldots,e_{sk_s}\}$ be a canonical basis of $V$ which is an $\mathbb{F}_q$-linear space of dimension $n$. Note that the set $B_i=\{e_{i_1},\ldots,e_{ik_i}\}$ forms a canonical basis of $V_i$.

Given $Q\subseteq P$, set:
$$V_Q=\{v\in V:supp_{\pi}(v)\subseteq Q\}.$$

\begin{Definition}
Let $\pi: [s]\rightarrow \mathbb{N}$ be a label and $P=([s],\leq)$ be a poset. An automorphism $\varphi\in Aut(P)$ is called a $(P,\pi)$-automorphism if, for all $i\in [s]$,
$$k_{\varphi(i)}=\pi(\varphi(i))=\pi(i)=k_i.$$
We denote by $Aut(P,\pi)$ the group of $(P,\pi)$-automorphisms.
\end{Definition}

\begin{theorem}\label{section}
Let $\varphi$ be a $(P,\pi)$-automorphism on $P$. Then the linear mapping $T_{\varphi}:V\rightarrow V$ given by
$$T_{\varphi}(e_{ij})=e_{\varphi(i)j}$$
is a $(P,\pi,w)$-isometry of $\left(V,d_{w,(P,\pi)}\right)$. Furthermore, the map
$$\Psi: Aut(P,\pi)\rightarrow GL_{w,(P,\pi)}(V)$$
defined by $\varphi\rightarrow T_{\varphi}$ is an injective group homomorphism.
\end{theorem}

\begin{proof}
Take $v=\sum\limits_{i=1}^s\sum\limits_{j=1}^{k_i}a_{ij}e_{ij}\in V$. Then $$T_{\varphi}(v)=\sum\limits_{i=1}^s\sum\limits_{j=1}^{k_i}a_{ij}T_{\varphi}(e_{ij})= \sum\limits_{i=1}^s\sum\limits_{j=1}^{k_i}a_{ij}e_{\varphi(i)j},$$
which implies that
$$supp_{\pi}(T_{\varphi}(v))=\left\{\varphi(i)\in P: a_{ij}\neq 0\ \text{for some}\ 1\leq j\leq k_i\right\}=\left\{\varphi(i)\in P:i\in supp_{\pi}(v)\right\}.$$
Therefore $I_{T_{\varphi}(v)}^P=\varphi(I_v^P)$. It follows from the fact that $\varphi$ is an order automorphism that $M_{T_{\varphi}(v)}^P=\varphi(M_v^P)$. Then
\begin{eqnarray*}
\overline{\omega}_{w,(P,\pi)}(T_{\varphi}(v))&=&\sum\limits_{i\in M_{T_{\varphi}(v)}^P}W_i(T_{\varphi}(v))+\sum\limits_{i\in I_{T_{\varphi}(v)}^p\setminus M_{T_{\varphi}(v)}^P}M_w\\
&=&\sum\limits_{i=\varphi(r),r\in M_v^P}W_i(T_{\varphi}(v))+\sum\limits_{i\in I_{T_{\varphi}(v)}^p\setminus M_{T_{\varphi}(v)}^P}M_w\\
&=&\sum\limits_{i=\varphi(r),r\in M_v^P}\max\{w(a_{rj}):1\leq j\leq k_r\}+\sum\limits_{i\in T_v^P\setminus M_v^p}M_w\\
&=&\sum\limits_{r\in M_v^P}W_r(v)+\sum\limits_{i\in T_v^P\setminus M_v^p}M_w\\
&=&\overline{\omega}_{w,(P,\pi)}(v).
\end{eqnarray*}
Hence $T_{\varphi}$ is a $(P,\pi,w)$-isometry. We next show that $\Psi$ is a group homomorphism. Let $\varphi$, $\psi$ be two $(P,\pi)$-automorphisms on $P$. One has
$$T_{\varphi\psi}(e_{ij})=e_{\varphi\psi(i)j}=T_{\varphi}(e_{\psi(i)j})= T_{\varphi}T_{\psi}(e_{ij})$$
and the map $\Psi$ is injective can be easily obtained by its definition.
\end{proof}

\begin{proposition}\label{T}
Let $T:V\rightarrow V$ be a linear isomorphism such that for every $i\in [s]$ and $v_i\in V_i$,
$$T(v_i)=u_i+\gamma_i$$
where $u_i\in V_i$, $\gamma_i\in V_{\langle i\rangle ^*}$ and $W_i(v_i)=W_i(u_i)$. Then $T$ is a $(P,\pi,w)$-isometry of $\left(V,d_{w,(P,\pi)}\right)$.
\end{proposition}

\begin{proof}
Let $v=v_1+v_2+\cdots+v_s\in V$ where $v_i\in V_i$. Then
$$T(v)=(u_1+\gamma_1)+\cdots+(u_s+\gamma_s)$$
where $\gamma_i\in V_{\langle i\rangle^*}$ and $u_i\in V_i$ such that $W_i(u_i)=W_i(v_i)$. Note that whenever $v_i\neq\textbf{0}$, we have $u_i\neq\textbf{0}$. Decompose $\gamma_i$ as
$$\gamma_i=\gamma_i^1+\cdots+\gamma_i^s$$
where $\gamma_i^l\in V_l$. As $\gamma_i\in V_{\langle i\rangle^*}$, $\gamma_i^l\neq \textbf{0}$ implies that $l<i$ in $P$. Then
$$T(v)=\sum\limits_{i=1}^s(u_i+(\gamma_1^i+\cdots+\gamma_s^i)).$$

Suppose that $i\in M_v^P\subseteq supp_{\pi}(v)$ and $\gamma_k^i\neq \textbf{0}$ for some $k\in[s]$. Then $k\in supp_{\pi}(v)$ and hence $i<k$ in $P$, a contradiction to the fact that $i$ is a maximal element of $supp_{\pi}(v)$. Therefore $\gamma_k^i=\textbf{0}$ for all $k\in[s]$ and the $i$-th component of $T(v)$ is
$$u_i+(\gamma_1^i+\cdots+\gamma_s^i)=u_i.$$
If $i\notin supp_{\pi}(T(v))$ then $u_i=\textbf{0}$. But $W_i(u_i)=W_i(v_i)=0$ implies that $v_i=\textbf{0}$, a contradiction. Hence $i\in M_{T(v)}^P\subseteq supp_{\pi}(T(v))$.

On the other hand, $M_{T(v)}^P\subseteq M_v^P$. In fact, if $j\in M_{T(v)}^P$ then the $j$-th component of $T(v)$ is
$$u_j+(\gamma_1^j+\cdots+\gamma_s^j).$$
If $\gamma_l^j\neq \textbf{0}$ then $l\in supp_{\pi}(v)$ and $j<l\leq i$ for some $i\in M_v^P\subseteq supp_{\pi}(T(v))$, a contradiction. Thus $\gamma_l^j=\textbf{0}$ for all $1\leq l\leq s$ and $u_j\neq \textbf{0}$. Note that $W_j(v_j)=W_j(u_j)\neq0$ implies that $v_j\neq\textbf{0}$. Therefore $j\in supp_{\pi}(v)$. If $j\notin M_v^P$, then $j<i$ for some $i\in M_v^P\subseteq supp_{\pi}(T(v))$, a contradiction. Hence $M_{T(v)}^P=M_v^P$ and thus $I_{T(v)}^P=I_v^P$.

As the $i$-th component of $T(v)$ is $u_i$ such that $W_i(v_i)=W_i(u_i)$ for all $i\in M_{T(v)}^P$, we have
\begin{eqnarray*}
\overline{\omega}_{w,(P,\pi)}(T(v))&=&\sum\limits_{i\in M_{T(v)}^P}W_i(T(v))+\sum\limits_{i\in I_{T(v)}^P\setminus M_{T(v)}^P}M_w\\
&=&\sum\limits_{i\in M_{T(v)}^P}W_i(u_i)+\sum\limits_{i\in I_{T(v)}^P\setminus M_{T(v)}^P}M_w\\
&=&\sum\limits_{i\in M_v^P}W_i(v_i)+\sum\limits_{i\in I_v^P\setminus M_v^P} M_w\\
&=&\overline{\omega}_{w,(P,\pi)}(v).
\end{eqnarray*}
This completes our proof.
\end{proof}

Let $B_i=\{e_{i1},e_{i2},\ldots,e_{ik_i}\}$ be a canonical basis of $V_i$. Then $\overline{\omega}_{w,(P,\pi)}(e_{i1})=\cdots=\overline{\omega}_{w,(P,\pi)}(e_{ik_i}) =w(1)+M_w\cdot |\langle i\rangle^*|$. We set $\overline{\omega}_{w,(P,\pi)}(B_i)=w(1)+M_w\cdot|\langle i\rangle^*|$.
Let $B=(B_{i_1},B_{i_2},\ldots,B_{i_s})$ be a total ordering of the basis of $V$ such that $B_{i_r}$ appears before $B_{i_l}$ wherever $\overline{\omega}_{w,(P,\pi)}(B_{i_r})\leq \overline{\omega}_{w,(P,\pi)}(B_{i_l})$ for all $i_r, i_l\in [s]$. Without loss of generality, we suppose that $B=\{B_1,B_2,\ldots,B_s\}$ is a total ordering basis of $V$. Then $|\langle r\rangle|<|\langle l\rangle|$ follows that all elements of $B_r$ come before $B_l$.

Let $\mathcal {T}$ be the set of mappings defined in Proposition \ref{T}.

\begin{corollary}
Let $B=(B_1,B_2,\ldots,B_s\}$ be a canonical base of $V$ defined as above. Given $T\in\mathcal {T}$, we have
$$T(e_{ij})=\sum\limits_{r\leq i}\sum\limits_{t=1}^{k_r}a_{ij}^{rt}e_{st}.$$
Moreover, $T$ can be represented by an $n\times n$ upper triangular block matrix with respect to $B$ as following
\begin{center}
  $[T]_B=\left[
     \begin{array}{ccccc}
       [T]_{B_1}^1 & [T]_{B_2}^1 & [T]_{B_3}^1&\cdots & [T]_{B_s}^1\\[1mm]
       O & [T]_{B_2}^2& [T]_{B_3}^2&\cdots & [T]_{B_s}^2 \\[1mm]
       O& O&[T]_{B_3}^3&\cdots&[T]_{B_s}^3\\[1mm]
       \vdots&\vdots&\vdots&\ddots&\vdots\\[1mm]
       O&O&O&\cdots&[T]_{B_s}^s
     \end{array}
   \right]$
\end{center}
where
\begin{center}
  $[T]_{B_r}^t=\left[
     \begin{array}{cccc}
     a_{r1}^{t1}& a_{r2}^{t1}&\cdots&a_{rk_r}^{t1}\\
     a_{r1}^{t2}& a_{r2}^{t2}&\cdots&a_{rk_r}^{t2}\\
     \vdots&\vdots&\ddots&\vdots\\
     a_{r1}^{tk_t}&a_{r2}^{tk_t}&\ldots&a_{rk_r}^{tk_t}\\
     \end{array}
   \right]$
\end{center}
and the element $v_{rl}=(a_{r_l}^{r_1},a_{r_l}^{r_2},\ldots,a_{rl}^{rk_r})\in V_r$ such that $W_r(v_{rl})=w(1)$ for all $1\leq l\leq k_r$, $W_r(\beta_1 v_{r1}+\cdots+\beta_{k_r}v_{rk_r})=W_r(\beta)$ where $\beta=(\beta_1,\ldots,\beta_{k_r})\in V_r$.
\end{corollary}

For each $i\in P$, the set $\langle i\rangle$ is an ideal and it is known as the \emph{principal ideal} generated by $i$.

\begin{theorem}\label{prime}
Let $T$ be a $(P,\pi,w)$-isometry of $\left(V,d_{w,(P,\pi)}\right)$. Then for every $i\in[s]$ and $\textbf{0}\neq v_i\in V_i$, we have that $\langle supp_{\pi}(T(v_i))\rangle$ is a principal ideal.
\end{theorem}

\begin{proof}
We present the proof in three steps.
\begin{itemize}
  \item \textbf{Step 1}

 We first prove that $\langle supp_{\pi}(T(e_{ij}))\rangle$ is a principal ideal. Let $\alpha\in\mathbb{F}_q$ such that $w(\alpha)=m_w$. Suppose that
$$T(\alpha e_{ij})=v_{r1}+\cdots+v_{rt}$$
where $\textbf{0}\neq v_{rl}\in V_{rl}$ for all $l\in [t]$. Then $supp_{\pi}(T(\alpha e_{ij}))=\{r1,\ldots,rt\}$. We now prove that there exists $k\in[t]$ such that $\overline{\omega}_{w,(P,\pi)}(v_{rk})=\overline{\omega}_{w,(P,\pi)}(T(\alpha e_{ij}))$. Assume the contrary, namely that for all $l\in [t]$, one has
$$\overline{\omega}_{w,(P,\pi)}(v_{rl})<\overline{\omega}_{w,(P,\pi)} \left(T(\alpha e_{ij})\right)=\overline{\omega}_{w,(P,\pi)}(\alpha e_{ij}).$$

It follows from $T^{-1}$ is a $(P,\pi,w)$-isometry of $\left(V,d_{w,(P,\pi)}\right)$ that
$${i}=supp_{\pi}(\alpha e_{ij})\subseteq \bigcup\limits_{l=1}^{t}supp_{\pi}\left(T^{-1}(v_{rl})\right),$$
which implies that $i\in supp_{\pi}\left(T^{-1}(v_{rk})\right)$ for some $k\in [t]$. Suppose that $W_i\left(T^{-1}(v_{rk})\right)=w(\beta)$ for some $\beta\in\mathbb{F}_q$, then
$$\overline{\omega}_{w,(P,\pi)}(\beta e_{ij})\leq \overline{\omega}_{w,(P,\pi)}\left(T^{-1}(v_{rk})\right)= \overline{\omega}_{w,(P,\pi)}(v_{rk})<\overline{\omega}_{w,(P,\pi)}(\alpha e_{ij}).$$
It follows that $w(\beta)+M_w\cdot |\langle i\rangle^*|<m_w+M_w\cdot|\langle i\rangle^*|$, a contradiction. Therefore, there exists $l\in [t]$ such that $\overline{\omega}_{w,(P,\pi)}(v_{rl})=\overline{\omega}_{w,(P,\pi)}(T(\alpha e_{ij}))$ and $I_{T(\alpha e_{ij})}^P=\langle rl \rangle$ is a principal ideal. Since $supp_{\pi}(T(\alpha e_{ij}))=supp_{\pi}(\alpha T(e_{ij}))=supp_{\pi}(T(e_{ij}))$, the result follows.

\item \textbf{Step 2}

  We next show that $I_{T(e_{i1})}^P=I_{T(e_{i2})}^P=\cdots=I_{T(e_{ik_i})}^P$. It is sufficient to show that $I_{T(e_{i1})}^P=I_{T(e_{i2})}^P$. Suppose that $I_{T(e_{i1})}^P=\langle k\rangle$ and $I_{T(e_{i2})}^P=\langle l\rangle$. Then
$$\left\{
  \begin{array}{ll}
    T(e_{i1})=u_k+\gamma_k & u_k\in V_k, \gamma_k\in V_{\langle k\rangle^*}, \\[2mm]
    T(e_{i2})=u_l+\gamma_l & u_l\in V_l, \gamma_l\in V_{\langle l\rangle^*}.
  \end{array}
\right.$$
Since $T$ preserves $(P,\pi,w)$-weight, we have
$$w(1)+M_w\cdot|\langle i\rangle^*|=\overline{\omega}_{w,(P,\pi)}(e_{i1})= \overline{\omega}_{w,(P,\pi)}(T(e_{i1}))= \overline{\omega}_{w,(P,\pi)}(u_k+\gamma_k)= W_k(u_k)+M_w\cdot|\langle k\rangle^*|.$$
It follows that $(W_k(u_k)-w(1))$ is divisible by $M_w$.
As $0< w(1), W_k(u_k)\leq M_w$, we have that $w(1)=W_k(u_k)$ and $|\langle i\rangle|=|\langle k\rangle|$. One can prove that $W_l(u_l)=w(1)$ and $|\langle i\rangle|=|\langle l\rangle|$ in similar way.

By the $(P,\pi,w)$-weight preservation and linearity of $T$,
\begin{eqnarray*}
w(1)+M_w\cdot|\langle i\rangle^*|&=&\overline{\omega}_{w,(P,\pi)}(e_{i1}+e_{i2})= \overline{\omega}_{w,(P,\pi)}(T(e_{i1})+T(e_{i2}))\\
&=& \overline{\omega}_{w,(P,\pi)}(u_k+\gamma_k+u_l+\gamma_l)\\
&=& \overline{\omega}_{w,(P,\pi)}(u_k+u_l).
\end{eqnarray*}
If $k\nleq l$, then
$$\overline{\omega}_{w,(P,\pi)}(u_k+u_l)\geq W_l(u_l)+M_w\cdot|\langle l\rangle^*|+ W_k(u_k)=2w(1)+ M_w\cdot|\langle i\rangle^*|,$$ a contradiction. Thus $k\leq l$. It follows from $|\langle k\rangle|=|\langle l\rangle|$ that $k=l$.

\item \textbf{Step 3}

Let $\textbf{0}\neq v_i=a_{i1}e_{i1}+\cdots+a_{ik_i}e_{ik_i}\in V_i$. From \textbf{Step 2}, we know that $I_{T(e_{i1})}^P=I_{T(e_{i2})}^P=\cdots=I_{T(e_{ik_i})}^P=\langle j\rangle$ for some $j$ such that $|\langle i\rangle|=|\langle j\rangle|$. Suppose that
$$T(e_{il})=u_l+\gamma_l,\ 1\leq l\leq k_i$$
where $\gamma_l\in V_{\langle j\rangle^*}$ and $u_l\in V_j$ such that $W_j(u_l)=w(1)$. Then
\begin{eqnarray*}
T(v_i)&=&a_{i1}T(e_{i1})+\cdots+a_{ik_i}T(e_{ik_i})\\
&=&a_{i1}(u_1+\gamma_1)+\cdots+a_{ik_i}(u_{k_i}+\gamma_{k_i})\\
&=&(a_{i1}u_1+\cdots+a_{ik_i}u_{k_i})+ (a_{i1}\gamma_1+\cdots+a_{ik_i}\gamma_{k_i}).
\end{eqnarray*}
If $a_{i1}u_1+\cdots+a_{ik_i}u_{k_i}=\textbf{0}\in V_j$ then $\overline{\omega}_{w,(P,\pi)}(T(v_i))\leq M_w\cdot|\langle j\rangle^*|$. On the other hand, $$\overline{\omega}_{w,(P,\pi)}(T(v_i))= \overline{\omega}_{w,(P,\pi)}(v_i)=W_i(v)+M_w\cdot|\langle i\rangle^*|>M_w\cdot|\langle i\rangle^*|=M_w\cdot|\langle j\rangle^*|,$$
a contradiction. Therefore $a_{i1}u_1+\cdots+a_{ik_i}u_{k_i}\neq\textbf{0}$ and hence $I_{T(v_i)}^P=\langle j\rangle$. This completes our proof.
\end{itemize}
\end{proof}

From the proof of Theorem \ref{prime}, we obtain a corollary.

\begin{corollary}\label{equal}
Let $T$ be a $(P,\pi,w)$-isometry of $\left(V,d_{w,(P,\pi)}\right)$. Then, for every $i\in [s]$, there is a unique $j\in[s]$ such that $T(V_i)\subseteq V_{\langle j\rangle}$ and  $|\langle i\rangle|=|\langle j\rangle|$. Moreover, for $\textbf{0}\neq v_i\in V_i$, there is a non-zero vector $u_j\in V_j$ and a vector $\gamma_j\in V_{\langle j\rangle^*}$ such that $T(v_i)=u_j+\gamma_j$ and $W_i(v_i)=W_j(u_j)$.
\end{corollary}

\begin{theorem}\label{order}
Let $T\in GL_{w,(P,\pi)}(V)$ and $i\leq j\in P$. Let $\textbf{0}\neq v_i\in V_i$ and let $\textbf{0}\neq v_j\in V_j$. Then
$$I_{T(v_i)}^P\subseteq I_{T(v_j)}^P.$$
\end{theorem}

\begin{proof}
It follows from Theorem \ref{prime} that $T_{T(v_i)}^P$ and $I_{T(v_j)}^P$ are principal ideals. By Corollary \ref{equal}, it is sufficient to show that $I_{T(e_{i1})}^P\subseteq I_{T(e_{j1})}^P$. Suppose that there are elements $k$ and $l$ in $P$ such that $I_{T(e_{i1})}^P=\langle k\rangle$ and $I_{T(e_{j1})}^P=\langle l\rangle$. Then
$$T(e_{i1})=u_k+\gamma_k\ \ \text{and}\ \ T(e_{j1})=u_l+\gamma_l$$
where $\gamma_t\in V_{\langle t\rangle^*}$ and $\textbf{0}\neq u_t\in V_t$ such that $W_t(u_t)=w(1)$ for $t=l,k$ respectively.

Note that $|\langle i\rangle|=|\langle k\rangle|$ and $|\langle j\rangle|=|\langle l\rangle|$. If $k=l$, there is nothing to prove. So we assume that $k\neq l$. This means that
$$\text{either}\ \  k\in supp_{\pi}(T(e_{i1})-T(e_{j1}))\ \ \text{or}\ \ l\in supp_{\pi}(T(e_{i1})-T(e_{j1})).$$
\begin{itemize}
  \item {\bf Case 1: } $k\notin supp_{\pi}(T(e_{i1})-T(e_{j1}))$. It follows that $k\in supp_{\pi}(T(e_{j1}))$ and hence $I_{T(e_{i1})}^P=\langle k\rangle\subseteq I_{T(e_{j1})}^P$.

  \item {\bf Case 2: } $l\notin supp_{\pi}(T(e_{i1})-T(e_{j1}))$. Then $l\in supp_{\pi}(T(e_{i1}))$ which implies that $l<k$. Therefore $I_{T(e_{j1})}^P=\langle l\rangle\subsetneq\langle k\rangle=I_{T(e_{i1})}^P$. Hence
\begin{eqnarray*} \overline{\omega}_{w,(P,\pi)}(e_{j1})&=&\overline{\omega}_{w,(P,\pi)}(T(e_{j1}))= \overline{\omega}_{w,(P,\pi)}(u_l+\gamma_l)\\
&=&W_l(u_l)+M_w\cdot|\langle l\rangle^*|=w(1)+M_w\cdot|\langle l\rangle^*|\\
&<&w(1)+M_w\cdot|\langle k\rangle^*|=W_k(u_k)+M_w\cdot|\langle k\rangle^*|\\
&=&\overline{\omega}_{w,(P,\pi)}(u_k)=\overline{\omega}_{w,(P,\pi)}(T(e_{i1}))\\
&=&\overline{\omega}_{w,(P,\pi)}(e_{i1}),
\end{eqnarray*}
which implies that $|\langle j\rangle|<|\langle i\rangle|$, a contradiction to the hypothesis that $i\leq j$.

  \item {\bf Case 3: } $k,l\in supp_{\pi}(T(e_{i1})-T(e_{j1}))$. For $1\leq t\leq s$, there exist $\alpha_t$, $\beta_t$ and $\zeta_t\in \mathbb{F}_q$ such that
$$w(\alpha_t)=W_t(T(e_{i1})),\  w(\beta_t)=W_t(T(e_{j1}))\ \text{and}\  w(\zeta_t)=W_t(T(e_{i1})-T(e_{j1})).$$
If the $l$-th component of $T(e_{i1})$ and the $k$-th component of $T(e_{j1})$ are both non-zeros. Then $l\leq k$ and $k\leq l$, a contradiction to the assumption that $k\neq l$. So, either $(T(e_{i1}))_l=\textbf{0}$ or $(T(e_{j1}))_k=\textbf{0}$ (here $(T(e_{i1}))_l$ denotes the $l$-th component of $T(e_{i1})$ and $(T(e_{j1}))_k$ denotes the $k$-th component of $T(e_{j1})$). By the $(P,\pi,w)$-weight preservation and the linearity of $T$,
\begin{eqnarray*}
\overline{\omega}_{w,(P,\pi)}(\zeta_ke_{k1}-\zeta_le_{l1})&\leq& \overline{\omega}_{w,(P,\pi)}(T(e_{i1})-T(e_{j1}))= \overline{\omega}_{w,(P,\pi)}(e_{i1}-e_{j1})\\
&\leq &\overline{\omega}_{w,(P,\pi)}(e_{j1})=\overline{\omega}_{w,(P,\pi)}(T(e_{j1}))= \overline{\omega}_{w,(P,\pi)}(u_l)\\
&=&w(1)+M_w\cdot|\langle l\rangle^*|.
\end{eqnarray*}
\begin{itemize}
  \item $(T(e_{i1}))_l=\textbf{0}$. Then $w(\zeta_l)=w(\beta_l)=W_l(T(e_{j1}))=W_l(u_l+\gamma_l)=W_l(u_l)=w(1)$ which implies that $k\leq l$ and hence $I_{T(e_{i1})}^P=\langle k\rangle\subseteq \langle l\rangle=I_{T(e_{j1})}^P$.
  \item $(T(e_{j1}))_k=\textbf{0}$. Then $w(\zeta_k)=w(\alpha_k)=w(1)$. If $k\nleq l$, then $$\overline{\omega}_{w,(P,\pi)}(\zeta_ke_{k1}-\zeta_le_{l1})\geq w(\zeta_k)+w(\zeta_l)+M_w\cdot|\langle l\rangle^*|=w(\zeta_l)+w(1)+M_w\cdot|\langle l\rangle^*|>w(1)+M_w\cdot|\langle l\rangle^*|,$$
  a contradiction.
\end{itemize}
\end{itemize}
\end{proof}

\begin{corollary}\label{down}
Given $T\in Gl_{w,(P,\pi)}(V)$ and $i\in[s]$, there is a unique $j\in[s]$ such that $|\langle i\rangle|=|\langle j\rangle|$ and $T\left(V_{\langle i\rangle}\right)\subseteq V_{\langle j\rangle}$.
\end{corollary}

\begin{proof}
Take $v\in V_{\langle i\rangle}$, then $v=v_{i_1}+\cdots+v_{i_t}$ where $i_l\leq i$ for each $l\in [t]$. It follows from Corollary \ref{equal} that there exists $j\in [s]$ such that $T(V_i)\subseteq V_{\langle j\rangle}$ and $I_{T(v_i)}^P=\langle j\rangle$ for $\textbf{0}\neq v_i\in V_i$. By Theorem \ref{order}, we have that $$I_{T(v_{i_l})}^P\subseteq I_{T(e_{i1})}^P=\langle j\rangle.$$
Hence
$$I_{T(v)}^P\subseteq\bigcup\limits_{l=1}^t I_{T(v_{i_l})}^P\subseteq \langle j\rangle.$$
\end{proof}

\begin{lemma}\label{dim}
Given $T\in GL_{w,(P,\pi)}(V)$. Let $i\in P$ and $j$ be the unique element of $P$ determined by $T(V_{\langle i\rangle})\subseteq V_{\langle j\rangle}$. Then $dim (V_i)=dim (V_j)$.
\end{lemma}

\begin{proof}
Consider a sequence of linear maps:
$$V_i\stackrel{T_i}{\longrightarrow} V_{\langle j\rangle}\stackrel{f}{\longrightarrow} V_{\langle j\rangle}/V_{\langle j\rangle^*}\stackrel{g}{\longrightarrow} V_j$$
where $T_i$ is a restriction of $T$ on $V_i$, $f$ is the canonical projection and $g$ is the isomorphism given by
$$u+V_{\langle j\rangle^*}=u_j.$$
Since
$ker(g\circ f\circ T_i)=\{\textbf{0}\}$,
we have that $g\circ f\circ T_i$ is injective and hence $dim(V_i)\leq dim(V_j)$.

We next show that $dim(V_j)\leq dim(V_i)$. It is sufficient to show that there exists $T^{'}\in GL_{w,(P,\pi)}(V)$ such that $T^{'}(V_j)\subseteq V_{\langle i\rangle}$. Let $\textbf{0}\neq v_i\in V_i$. Then $T(v_i)=v_j+\gamma_j$ where $\textbf{0}\neq v_j\in V_j$ such that $W_i(v_i)=W_j(v_j)$ and $\gamma_j\in V_{\langle j\rangle^*}$. By the linearity of $T^{-1}$, we have $T^{-1}(v_j)=v_i-T^{-1}(\gamma_j)$. Note that $i\notin T^{-1}(\gamma_j)$, otherwise
$$\overline{\omega}_{w,(P,\pi)}(\gamma_j)= \overline{\omega}_{w,(P,\pi)}\left(T^{-1}(\gamma_j)\right)\geq m_w+M_w\cdot|\langle i\rangle^*|=m_w+M_w\cdot|\langle j\rangle^*|>M_w\cdot|\langle j\rangle^*|,$$
a contradiction to the fact that $\gamma_j\in V_{\langle j\rangle^*}$. Therefore $i\in supp_{\pi}\left(T^{-1}(v_j)\right)$. As
$$\overline{\omega}_{w,(P,\pi)}\left(T^{-1}(v_j)\right)= \overline{\omega}_{w,(P,\pi)}(v_j)=W_j(v_j)+M_w\cdot|\langle j\rangle^*|=W_i(v_i)+M_w\cdot\left|\langle i\rangle^*\right|,$$
which implies that $\left\langle supp_{\pi}\left(T^{-1}(v_j)\right)\right\rangle=\langle i\rangle$. By Corollary \ref{equal}, we have $T^{-1}(V_j)\subseteq V_{\langle i\rangle}$. This completes our proof.
\end{proof}

\begin{lemma}
Let $T$ be a $(P,\pi,w)$-isometry of $\left(V,d_{w,(P,\pi)}\right)$. Define the map $\eta_T: P\rightarrow P$ as
$$\eta_T(i)=\max I_{T(v_i)}^P,$$
where $\textbf{0}\neq v_i\in V_i$. Then $\eta_T$ is a $(P,\pi)$-automorphism of $P$.
\end{lemma}

\begin{proof}
By Corollary \ref{equal}, $\eta_T$ is well-defined. Theorem \ref{order} assures that $\eta_T$ is an order preserving map.

It remains to show that $\eta_T$ is one-to-one. Suppose that $\eta_T(i)=\eta_T(j)=l$. Lemma \ref{dim} assures that $k_i=k_l=k_j$. Let $\alpha\in\mathbb{F}_q$ such that $w(\alpha)=M_w$.
By Corollary \ref{equal}, we have
$$\left\{
  \begin{array}{ll}
    T(\alpha e_{i1})=u_l+\gamma_l & u_l\in V_l, \gamma_l\in V_{\langle l\rangle^*}, \\[2mm]
    T(\alpha e_{j1})=u_l^{'}+\gamma_l^{'} & u_l^{'}\in V_l, \gamma_l^{'}\in V_{\langle l\rangle^*},
  \end{array}
\right.$$
and $|\langle i\rangle|=|\langle l\rangle|=|\langle j\rangle|$. Furthermore, $W_l(u_l)=W_l(u_l^{'})=w(\alpha)=M_w$.
Then
\begin{eqnarray*}
\overline{\omega}_{w,(P,\pi)}(\alpha e_{i1}+\alpha e_{j1})&=&\overline{\omega}_{w,(P,\pi)}\left(T(\alpha e_{i1})+T(\alpha e_{j1})\right)\\
&=&\overline{\omega}_{w,(P,\pi)}(u_l+u_l^{'})=W_l(u_l+u_l^{'})+M_w\cdot|\langle l\rangle^*|\\
&\leq& M_w+M_w\cdot|\langle l\rangle^*|=W_l(u_l)+M_w\cdot|\langle l\rangle^*|\\
&=&\overline{\omega}_{w,(P,\pi)}(u_l)=\overline{\omega}_{w,(P,\pi)}(T(\alpha e_{i1}))\\
&=&\overline{\omega}_{w,(P,\pi)}(\alpha e_{i1})
\end{eqnarray*}
Therefore $j\leq i$. Similarly, one can prove that $i\leq j$ and hence $i=j$. Therefore $\eta_T\in Aut(P,\pi)$.
\end{proof}

Recall that $\mathcal {T}$  is the set of mappings defined in Proposition \ref{T}. In the following, we will show that $\mathcal {T}$ is a normal subgroup of $GL_{w,(P,\pi)}(V)$.

\begin{lemma}
Consider the map $\Lambda: GL_{w,(P,\pi)}(V)\rightarrow Aut(P,\pi)$ given by
$$\Lambda(T)=\eta_T.$$
Then $\Lambda$ is a surjective group homomorphism and $ker(\Lambda)=\mathcal {T}$.
\end{lemma}

\begin{proof}
Let $T$, $T^{'}\in GL_{w,(P,\pi)}(V)$, $i\in [s]$ and $\textbf{0}\neq v_i\in V_i$. Suppose that $\eta_T(i)=j$ and $\eta_{T^{'}}(j)=k$. Take $v_i\in V_i$ such that $W_i(v_i)=m_w$. Then $T(v_i)=v_j+\gamma_j$ where $\gamma_j\in V_{\langle j\rangle^*}$ and $v_j\in V_j$ such that $W_j(v_j)=m_w$. Meanwhile, $T^{'}(v_j)=v_k+\gamma_k$ where $\gamma_k\in V_{\langle k\rangle^*}$ and $v_k\in V_k$ satisfies $W_k(v_k)=W_j(v_j)=m_w$. We can see that
$$T^{'}T(v_i)=T^{'}(v_j+\gamma_j)=v_k+\gamma_k+T^{'}(\gamma_j).$$
Since $T$ and $T^{'}$ preserve $(P,\pi)$-weight, we have
\begin{eqnarray*}
\overline{\omega}_{w,(P,\pi)}\left(T^{'}(\gamma_j)\right)&= &\overline{\omega}_{w,(P,\pi)}(\gamma_j)<\overline{\omega}_{w,(P,\pi)}(v_j)= \overline{\omega}_{w,(P,\pi)}\left(T^{'}(v_j)\right)\\
&=&\overline{\omega}_{w,(P,\pi)}(v_k+\gamma_k)=\overline{\omega}_{w,(P,\pi)}(v_k)\\
&=&m_w+M_w\cdot|\langle k\rangle^*|,
\end{eqnarray*}
which implies that $T^{'}(\gamma_j)\in V_{\langle k\rangle^*}$. Therefore $T^{'}T(v_i)=v_k+u_k$ where $u_k\in V_{\langle k\rangle^*}$. It follows that  $\eta_{T^{'}}\eta_T(i)=k=\eta_{T^{'}T}(i)$.

Consider the map $\Psi: Aut(P,\pi)\rightarrow GL_{w,(P,\pi)}(V)$ defined in Theorem \ref{section}, we can see that $\Lambda\circ\Psi= Id_{P}$ which implies that $\Lambda$ is surjective.

It remains to show that $ker(\Lambda)=\mathcal {T}$. It is clear that $\mathcal {T}\subseteq ker(\Lambda)$. On the other hand, if $T\in ker(\Lambda)$, then $T(V_i)\subseteq V_{\langle i\rangle}$ for every $i\in[s]$. Let $\textbf{0}\neq v_i\in V_i$. Suppose that $T(v_i)=u_i+\gamma_i$ where $u_i\in V_i$ and $\gamma_i\in V_{\langle i\rangle^*}$, then $$\overline{\omega}_{w,(P,\pi)}(v_i)=\overline{\omega}_{w,(P,\pi)}(T(v_i))= \overline{\omega}_{w,(P,\pi)}(u_i+\gamma_i)>M_w\cdot|\langle i\rangle^*|$$
 which implies that $u_i\neq \textbf{0}$. Furthermore, we have $W_i(v_i)=W_i(u_i)$.
\end{proof}

Denote by $\mathcal {A}=\Psi(Aut(P,\pi))$. It follows from Theorem \ref{section} that $\mathcal {A}\cong Aut(P,\pi)$.

Let $M_{n\times n}(\mathbb{F}_q)$ denotes the set of all $n\times n$ matrix over $\mathbb{F}_q$ and
$$U_{w,(P,\pi)}=\left\{\left(A_{ij}\right)\in M_{n\times n}(\mathbb{F}_q):
\left.
\begin{array}{l}
A_{ij}\in M_{k_i\times k_j}(\mathbb{F}_q);
A_{ij}=O\ \text{if}\ i\nleq j;\\[2mm]
W_r(v_{rl})=w(1)\ \text{where}\ v_{rl}\ \text{is the $l$-th column in}\ A_{rr},\ \text{and}\\[2mm]
W_r(\beta_1v_{r1}+\cdots+\beta_{kr}v_{rk_r})=W_r(\beta)\ \text{where}\ \beta=(\beta_1,\ldots,\beta_{kr})\in V_r
\end{array}
\right.\right\}.$$

\begin{theorem}\label{semi}
With the notations above, we have
$$GL_{w,(P,\pi)}(V)\cong \mathcal {T}\rtimes\mathcal {A}\cong U_{w,(P,\pi)}\rtimes Aut(P,\pi).$$
\end{theorem}

\begin{proof}
Let $T$ be a $(P,\pi,w)$-isometry of $\left(V,d_{w,(P,\pi)}\right)$. Let $\eta=\eta_T$,
then $S\circ (T_{\eta})^{-1}=S\circ T_{\eta^{-1}}\in\mathcal {T}$ and
$$S=(S\circ (T_{\eta^{-1}})\circ T_{\eta}.$$

It remains to show that $\mathcal {T}\cap\mathcal {A}=\{Id\}$ where $Id$ is the identity mapping of $(V,d_{w,(P,\pi)})$. Suppose that $T\in \mathcal {T}\cap\mathcal {A}$. Since $T\in\mathcal {T}=ker(\Lambda)$, we have $\Lambda(T)=Id_{P}$. On the other hand $T=\Psi(\varphi)$ for some $\varphi\in Aut(P,\pi)$ and $\Lambda\Psi(\varphi)=\varphi=Id_P$. Hence $T=\Psi(Id_P)=Id$.
\end{proof}

\section{Examples}

\quad\; Here we will illustrate our conclusion with some examples.

\begin{example}
When the label $\pi$ satisfies $\pi(i)=1$ for all $i\in [s]$, the weighted poset block metric induces weighted poset metric, $d_{(P,w)}$, introduced in [\ref{panek}]. Immediate substitution gives that
$$U_{w,(P,\pi)}=\left\{(a_{ij})\in M_{s\times s}(\mathbb{F}_q):
\left.
\begin{array}{l}
a_{ij}=0\ \text{if}\ i\nleq j\ \text{and}\\[2mm]
w(a_{ii})=w(1)\ \text{such that}\ w(\alpha a_{ii})=w(\alpha)\ \text{where}\ \alpha\in\mathbb{F}_q
\end{array}
\right.\right\}.$$
Then the characterization of $\left(\mathbb{F}_q^s,d_{w,(P,\pi)}\right)$ given in [\ref{panek}, Theorem 19]  follows from Theorem \ref{semi} as a particular case:
$$GL_{w,(P,\pi)}(\mathbb{F}_q^s)\cong U_{w,(P,\pi)}\rtimes Aut(P).$$
\end{example}

\begin{example}
When the weight $w$ on $\mathbb{F}_q$ is the Hamming weight $w_H$, the weighted poset block metric induces poset block metric introduced in [\ref{ALVES}]. Immediate substitution gives that
$$U_{w_H,(P,\pi)}=\left\{(A_{ij})\in M_{n\times n}(\mathbb{F}_q):
\left.
\begin{array}{l}
A_{ij}\in M_{k_i\times k_j}(\mathbb{F}_q);\\[2mm]
A_{ij}=O\ \text{if}\ i\nleq j;\\[2mm]
A_{ii}\ \text{is invertible}
\end{array}
\right.\right\}.$$
Then, the characterization of the group of linear isometries of $\left(\mathbb{F}_q^n,d_{w,(P,\pi)}\right)$ follows from Theorem \ref{semi} as a particular case:
$$GL_{w_H,(P,\pi)}(V)\cong U_{w_H,(P,\pi)}\rtimes Aut(P,\pi).$$
\end{example}

We remark that the results in Section 3 are also valid when we consider the Lee weight $w_L$ over $\mathbb{Z}_m$ instead of $\mathbb{F}_q$.

\begin{example}
When the weight $w$ on $\mathbb{Z}_m$ is the Lee weight $w_{L}$, the weighted poset block metric induced pomset block metric. Immediate substitution gives that
$$U_{w_L,(P,\pi)}=\left\{(A_{ij})\in M_{n\times n}(\mathbb{Z}_m)£º
\left.
\begin{array}{l}
A_{ij}\in M_{k_i\times k_j}(\mathbb{Z}_m);
A_{ij}=O\ \text{if}\ i\nleq j;\\[2mm]
W_r(v_{rl})=\pm 1\ \text{where}\ v_{rl}\ \text{is the $l$-th column in}\ A_{rr},\ \text{and}\\[2mm]
W_r(\beta_1v_{r1}+\cdots+\beta_{kr}v_{rk_r})=W_r(\beta)\ \text{where}\ \beta=(\beta_1,\ldots,\beta_{kr})\in V_r
\end{array}
\right.\right\}.$$
Then, the characterization of the group of linear isometries of $\left(\mathbb{Z}_m^n,d_{w_L,(P,\pi)}\right)$ follows from Theorem \ref{semi} as a particular case:
$$GL_{w_L,(P,\pi)}(V)\cong U_{w_L,(P,\pi)}\rtimes Aut(P,\pi).$$
\end{example}

\begin{example}
When $P$ is an anti-chain and the label $\pi$ satisfies
\begin{center}
$k_1=\ldots k_{t_1}=m_1$\\
$k_{t_1+1}=\ldots=k_{t_1+t_2}=m_2$\\
$\vdots$\\
$k_{t_1+t_2+\cdots+t_{l-1}+1}=\ldots=k_n=m_{l},$
\end{center}
we have
$$Aut(P,\pi)\cong S_{t_1}\times S_{t_2}\times \cdots \times S_{t_l}.$$
In this case $\overline{\omega}_{w,(P,\pi)}(v)=\sum\limits_{i\in supp_{\pi}(v)}W_i(v)$. Denote by $GL(k_i,w)$ the group of the linear transformation $T:\mathbb{F}_q^{k_i}\rightarrow \mathbb{F}_q^{k_i}$ such that $\overline{\omega}_{w,(P,\pi)}(T(v_i))=\overline{\omega}_{w,(P,\pi)}(v_i)$. Then we have
$$\mathcal {T}\cong GL(k_1,w)\times GL(k_2,w)\times\cdots\times GL(k_n,w).$$
It follows from Theorem \ref{semi} that
$$GL_{w,(P,\pi)}(V)\cong \left(\prod\limits_{i=1}^n GL(k_i,w)\right)\rtimes \left(\prod\limits_{i=1}^l S_{t_i}\right).$$
Note that when $\pi(i)=1$, $Aut(P,\pi)\cong  S_n$ and $GL(k_i,w)$ is the group of linear transformation $T: \mathbb{F}_q\rightarrow\mathbb{F}_q$ that preserves the weight $w$. Then we have
$$GL_{w,(P,\pi)}(V)\cong \left(\prod\limits_{i=1}^n GL(1,w)\right)\rtimes S_n.$$
\end{example}

\section{MDS codes and perfect codes in weighted poset block metric}

\subsection{Basic definitions}
\quad\; In this section, the notation of $P$, $\pi$, $w$ and $V$ is the same as Section 3.  A subset $C\subseteq V$ with cardinality $K$ is said to be an \emph{$(n,K,d_w(C))$ $(P,\pi,w)$-code}, where $V$ is endowed with the weighted poset block metric $d_{w,(P,\pi)}(.,.)$ and
$$d_w(C)=\min\{d_{w,(P,\pi)}(u,v): u\neq v\in C\}$$
is the minimal $(P,\pi,w)$-distance of $C$.
When the weight $w$ over $\mathbb{F}_q$ is considered to be Hamming weight $w_H$, we denote by $d_H(C)=d_{w_H}(C)$. A \emph{linear $(P,\pi,w)$-code} is a subspace of $V$.

Denote by $\mathcal {I}(P)$ the set of all ideals of $P$. Now we state the Singleton bound for the case of weighted poset block metric over $\mathbb{F}_q^n$.

\begin{theorem}\label{SIN}(Singleton Bound)
Let $C$ be an $(n,K,d_w(C))$ $(P,\pi,w)$-code. Let $\lambda=\left\lfloor\frac{d_w(C)-m_w}{M_w}\right\rfloor$ and $\mu=\max\limits_{I\in \mathcal {I}(P),|I|=\lambda}\sum\limits_{i\in I}k_i$. Then
$$K\leq q^{n-\mu}.$$
\end{theorem}

\begin{proof}
Let $I\in \mathcal {I}(P)$ with $|I|=\lambda$. Let $u$ and $v$ be two distinct elements of $C$. If $u$ and $v$ coincide in all positions out of $I$, then
$$d_{w,(P,\pi)}(u,v)=\overline{\omega}_{w,(P,\pi)}(u-v)\leq|I|\cdot M_w=\left\lfloor\frac{d_w(C)-m_w}{M_w}\right\rfloor\cdot M_w<d_w(C),$$
a contradiction. This means that any two distinct codewords of $C$ will differ in at least one position outside $I$. Therefore there exists an injective map from $C$ to $\mathbb{F}_q^{n-\sum\limits_{i\in I}k_i}$ which implies that $K\leq q^{n-\sum\limits_{i\in I}k_i}$. Hence $K\leq q^{n-\sum\limits_{i\in I}k_i}$ for any $I\in\mathcal {I}(P)$ with $|I|=\lambda$.
\end{proof}

\begin{remark}
Let $P$ be a chain with usual order $1<2<\cdots<s$ and $\pi(i)=1$ for all $i\in[s]$. For a linear $(n,q^k,d_w(C))$ $(P,\pi,w)$-code $C$, Theorem \ref{SIN} deduces
$$d_w(C)\leq M_w(n-k)+m_w,$$
which is the same as the Singleton bound for the case of weighted poset metric when $C$ is a linear code and $P$ is a chain [\ref{panek}, Corollary 24].
\end{remark}

\begin{remark}
Let $C$ be an $(n,K,d_w(C))$ $(P,\pi,w)$-code with $K=q^k$. When the weight $w$ is taken to be Hamming weight over $\mathbb{F}_q$, Theorem \ref{SIN} deduces
$$n-k \geq \max\limits_{I\in\mathcal {I}(P),|I|=d_H(C)-1}\sum\limits_{i\in I} k_i,$$
which is the same as the Singleton bound for the case of poset block metric when $C$ is a linear code [\ref{DASS}, Theorem 3.2]. Note that our results holds for non-linear codes as well.
\end{remark}

\begin{remark}
Let $P$ be a chain. Without loss of generality, we may assume that $P$ has the chain order $1<2<\cdots<s$. For an $(n,K,d_w(C))$ $(P,\pi,w)$-code $C$. Theorem \ref{SIN} deduces
$$K\leq q^{n- (k_1+k_2+\cdots+k_{\lambda})}.$$
\end{remark}

We now define a maximum distance separable $(P,\pi,w)$-code.

\begin{Definition}
An $(n,K,d_w(C))$ $(P,\pi,w)$-code $C\subseteq V$ is called a maximum distance separable (MDS) $(P,\pi,w)$-code if it attains the Singleton bound.
\end{Definition}

Let $P=\left([s],\leq_P\right)$ and $Q=\left([s],\leq_Q\right)$ be two posets. We say that $Q$ is \emph{finer} than $P$ if $i\leq_P j$ in $P$ implies that $i\leq_Q j$ in $Q$.

\begin{lemma}\label{finer}
Let $P=\left([s],\leq_P\right)$ and $Q=\left([s],\leq_Q\right)$ be two posets such that $Q$ is finer than $P$. Let $\pi$ be a labeling such that $\pi(1)=\pi(2)=\cdots=\pi(s)=t$. If $C$ is an MDS $(P,\pi,w)$-code, then $C$ is an MDS $(Q,\pi,w)$-code.
\end{lemma}

\begin{proof}
 Let $C$ be an MDS $(P,\pi,w)$-code with diameters $(n,K,d_w(C))$. Then $C$ is an $(n,K,d_w'(C))$ $(Q,\pi,w)$-code (here $d_w'(C)$ is the minimal $(Q,\pi,w)$-distance of $C$). Since $Q$ is finer than $P$, we have $d_{w,(P,\pi)}(u,v)\leq d_{w,(Q,\pi)}(u,v)$ for $u,v\in V$ which implies that $d_w(C)\leq d_w'(C)$. Then
 $$\lambda_1=\left\lfloor\frac{d_w(C)-m_w}{M_w}\right\rfloor\leq \left\lfloor\frac{d_w'(C)-m_w}{M_w}\right\rfloor=\lambda_2$$
 and hence
 $$\mu_1=\max\limits_{I\in \mathcal {I}(P),|I|=\lambda_1}\sum\limits_{i\in I}k_i=\lambda_1t\leq\lambda_2t=\max\limits_{I\in \mathcal {I}(Q),|I|=\lambda_2}\sum\limits_{i\in I}k_i=\mu_2.$$
 Therefore
 $$K=q^{n-\mu_1}\geq q^{n-\mu_2}\geq K.$$
 This forces that $K=q^{n-\mu_2}$. Therefore $C$ is an MDS $(Q,\pi,w)$-code.
\end{proof}

Let $P=([s],\leq)$ be an anti-chain and $Q=([s],\leq)$ be a poset. Then $Q$ is finer than $P$. When the weight $w$ on $\mathbb{F}_q$ is taken to be the Hamming weight $w_H$, we get the following result that appeared in [\ref{DASS}]

\begin{corollary}
If a code equipped with error-block metric is an MDS code, then it is also an MDS poset block code for every partial order defined on the set $[s]$.
\end{corollary}

When $\pi(i)=1$ for all $i\in[s]$ and the weight on $\mathbb{Z}_m$ is taken to be the Lee weight $w_L$, we can get the following result has been represented in [\ref{IGS3}].

\begin{corollary}
Let $C\subseteq \mathbb{Z}_m^n$ be an $(n,K)$ code. Let $\mathbb{P}=(M,R)$ be a pomset where $M=\left\{\lfloor\frac{m}{2}\rfloor/1,\ldots,\lfloor\frac{m}{2}\rfloor/n\right\}$ and $R$ is a pomset relation on $M$. If a code $C$ is an MDS code with Lee metric, then it is MDS pomset code for any pomset relation defined on $M$.
\end{corollary}

\begin{Definition}
Let $w$ be a weight on $\mathbb{F}_q$. For $u\in \mathbb{F}_q^n$, the \emph{$(P,\pi,w)$-ball} with center $u$ and radius $r$ is the set
$$B_{w,(P,\pi)}(u,r)=\{v\in\mathbb{F}_q^n: d_{w,(P,\pi)}(u,v)\leq r\}.$$
When the wight $w$ over $\mathbb{F}_q$ is considered to be Hamming weight, we denote by $B_{(P,\pi)}(u,r)$ the $(P,\pi,w)$-ball with center $u$ and radius $r$.
\end{Definition}

\begin{Definition}
A code $C$ is said to be an \emph{$r$-perfect $(P,\pi,w)$-code} if the $(P,\pi,w)$-balls of radius $r$ centered at the codewords of $C$ are pairwise disjoint and their union is $V$.
\end{Definition}

\subsection{Weighted poset block metric with chain poset}
\quad\;In what follows, we always assume that $P$ is a chain defined by $1<2<\cdots<s$. Recall the notations given in (\ref{sum}) and (\ref{decom}), we have the following.

\begin{theorem}\label{perfect}
Let $C\subseteq V$ be a $(P,\pi,w)$-code and $r=t\cdot M_w$. Then $C$ is $r$-perfect if and only if there is a function
$$f: \bigoplus\limits_{j=t+1}^{s}V_j\rightarrow\bigoplus\limits_{i=1}^{t}V_i$$
such that
$$C=\left\{(f(v),v):v\in\bigoplus\limits_{j=t+1}^{s}V_j\right\}.$$
\end{theorem}

\begin{proof}
Assume that $C$ is $r$-perfect. Let $v\in \bigoplus\limits_{j=t+1}^{s}V_j$. As $C$ is a $r$-perfect code, there exists $c\in C$ such that $(\textbf{0},v)\in B_{w,(P,\pi)}(c,r)$ which implies that $c-(\textbf{0},v)\in B_{w,(P,\pi)}(0,r)$. Since $P$ is a chain, we have
$$B_{w,(P,\pi)}(0,r)=\{(x_1,\ldots,x_t,0,\ldots,0):x_i\in V_i\}.$$
Therefore there exists $u\in\bigoplus\limits_{i=1}^{t}V_i$ such that $c-(\textbf{0},v)=(u,\textbf{0})$ and hence $c=(u,v)$. Suppose there exists $c'\neq c\in C$ such that $(\textbf{0},v)\in B_{w,(P,\pi)}(c')$. Then $c'=(u',v)$ and $c-c'=(u-u',\textbf{0})\in B_{w,(P,\pi)}(0,r)$, a contradiction to the hypothesis that $C$ is an $r$-perfect code. Thus we can define a function $f: \bigoplus\limits_{j=t+1}^{s}V_j\rightarrow\bigoplus\limits_{i=1}^{t}V_i$ which sends $v\in \bigoplus\limits_{j=t+1}^{s}V_j$ to the unique $u\in\bigoplus\limits_{i=1}^{t}V_i$ such that $c=(f(v),v)$.

On the contrary, suppose that there exists a function $f$ such that $C=\left\{(f(v),v):v\in\bigoplus\limits_{j=t+1}^{s}V_j\right\}$. Then
$$B_{w,(P,\pi)}((f(v),v),r)=\left\{(u,v): u\in\bigoplus\limits_{i=1}^{t}V_i\right\}$$
and thus $|B_{w,(P,\pi)}(c,r)|=q^{k_1+k_2+\cdots+k_t}$. Furthermore, for any $(u,v)\in V$, one has $C\cap B_{w,(P,\pi)}((u,v),r)=(f(v),v)$. Therefore $\mathop{\bigcup}\limits_{c\in C}B_{w,(P,\pi)}(c,r)$ is a disjoint union and its order is $|C|\cdot q^{k_1+k_2+\cdots+k_t}=q^n$. Thus, $C$ is $r$-perfect.
\end{proof}

\begin{remark}
If $C$ is a linear $(P,\pi,w)$-code, then the function $f$ given in Theorem \ref{perfect} is a linear map.
\end{remark}

\begin{proof}
Let $u,v\in \bigoplus\limits_{j=t+1}^{s}V_j$. Then $(f(u),u),(f(v),v)\in C$. Since $C$ is a linear code, we have that $\alpha(f(u),u)+\beta(f(v),v)=(\alpha f(u)+\beta f(v),\alpha u+\beta v)\in C$ for $\alpha,\beta\in\mathbb{F}_q$ which implies that $f(\alpha u+\beta v)=\alpha f(u)+\beta f(v)$. Therefore $f$ is a linear map.
\end{proof}

\begin{Definition}
The \emph{packing radius} $\rho(C)$ of a code $C$ is the largest radius of spheres centered at codewords so that the spheres are pairwise disjoint. We call a code $C$ is \emph{perfect} if it is $\rho(C)$-perfect.
\end{Definition}

The following proposition is a generalization of [\ref{panek}, Lemma 21, Corollary 22] wherein the metric considered was weighted poset metric. The proof is on similar lines and hence omitted.

\begin{proposition}\label{radius}
Let $r=l+i\cdot M_w$ where $l\in[M_w]$ and $i\geq 0$ is an integer. Let $v\in V$ and $C$ be a $(P,\pi,w)$-code. Then
\begin{enumerate}[(1)]
\item $B_{w,(P,\pi)}(v,r)\subseteq B_{(P,\pi)}(v,i+1)$. Moreover, $B_{w,(P,\pi)}(v,r)= B_{(P,\pi)}(v,i+1)$ if and only if $l=M_w$.
\item $\rho(C)\geq M_w\cdot (d_{H}(C)-1)$. Moreover, $\rho(C)=\left(d_{H}(C)-1\right)M_w$ if and only if $d_{w,(P,\pi)}(C)=m_w+\left(d_H(C)-1\right)M_w$.
\end{enumerate}
\end{proposition}

\begin{theorem}\label{PERFECT}
Let $C$ be a $(P,\pi,w)$-code such that $d_{w,(P,\pi)}(C)=m_w+\left(d_H(C)-1\right)M_w$. Then $C$ is MDS if and only if $C$ is perfect.
\end{theorem}

\begin{proof}
Suppose that $C$ is an MDS code. Then $|C|=q^{n-(k_1+k_2+\cdots+k_{\lambda})}$, where $\lambda=\left\lfloor\frac{d_w(C)-m_w}{M_w}\right\rfloor$. Denote by $d=d_H(C)$. Since $P$ is a chain, we have that
$$d_w(C)=m_w+M_w\cdot (d-1).$$
Therefore
$$\lambda=\left\lfloor\frac{d_w(C)-m_w}{M_w}\right\rfloor=\left\lfloor \frac{m_w+M_w\cdot (d-1)-m_w}{M_w}\right\rfloor=d-1.$$

From the proof of Theorem \ref{SIN}, we have that there exists an injective map
$$g:\mathbb{F}_q^{n-(k_1+k_2+\cdots+k_{\lambda})}\rightarrow \mathbb{F}_q^{k_1+k_2+\cdots+k_{\lambda}}$$
and
$$C=\left\{\left(g(u),u\right): u\in \mathbb{F}_q^{k_{\lambda+1}+\cdots+k_s}\right\}.$$
It follows from Theorem \ref{perfect} and Proposition \ref{radius} (2) that $C$ is $\lambda M_w$-perfect, that is, $C$ is $\rho(C)$-perfect.

Conversely, if $C$ is a perfect code, then
$|C|\cdot |B_{w,(P,\pi)}(0,\rho(C))|=q^n$. It follows from Proposition \ref{radius} (1) that
$$|B_{w,(P,\pi)}(0,\rho(C))|=|B_{w,(P,\pi)}(0,M_w\cdot(d-1))|= |B_{(P,\pi)}(0,d-1)|=q^{k_1+\cdots+k_{d-1}}.$$
Therefore
$$|C|=q^{n-(k_1+\cdots+k_{d-1})}=q^{n-(k_1+\cdots+k_{\lambda})}$$
which implies that $C$ is an MDS $(P,\pi,w)$-code.
\end{proof}

As a direct application, if we choose the weight $w$ on $\mathbb{F}_q$ as the Hamming weight $w_H$, we immediately get the following result that appeared in [\ref{DASS}].

\begin{corollary}
A poset block code $C$ is perfect if and only if $C$ is MDS with poset block metric.
\end{corollary}

Let $\pi$ be a labeling satisfing $\pi(i)=1$ for every $i\in[s]$. Applying Theorem \ref{PERFECT}, we can get the following result  which has represented in [\ref{ANTI}] and [\ref{panek}].

\begin{corollary}
Let $C$ be a code with weighted poset metric. Then $C$ is MDS if and only if $C$ is perfect.
\end{corollary}

\begin{remark}
Note that Theorem \ref{PERFECT} is only applicable for the case when $P$ is considered to be a chain. It is not applicable for Hamming metric and Lee metric.
\end{remark}

\section{Conclusion}

\quad\; In this paper, we study weighted poset block metric over $\mathbb{F}_q^n$ which is a generalization to metrics such as Hamming metric, Lee metric, poset metric, pomset metric, poset block metric and error-block metric and so on. We give a complete description of the groups of linear isometries of weighted poset block space in terms of a semi-direct product of its two subgroups. Note that our conclusion remains valid if we replace $V=\bigoplus\limits_{i=1}^{s}\mathbb{F}_q^{k_i}$ with $V =\bigoplus\limits_{i=1}^{s}R^{k_i}$ where $R$ is an associative ring with identity and there exists a multiplicative invertible element $\alpha\in R$ such that $w(\alpha)=m_w$. Moreover, basic parameters such as packing radius and bounds for minimum distance of weighted poset block codes are established and the relationship between MDS codes and perfect codes when the poset is considered to be a chain is investigated immediately.

\end{document}